\documentclass{amsart}
\usepackage{amsfonts,amssymb,amscd,amsmath,enumerate,verbatim,calc}
\newtheorem{theorem}{Theorem}[section]

\theoremstyle{definition}
\theoremstyle{definitions}
\newtheorem{definition}[theorem]{Definition}

\newtheorem{remark}[theorem]{Remark}

\theoremstyle{notations}
\theoremstyle{note}

\theoremstyle{remarks}

\newcommand{\T}{\mathrm}

\newcommand{\F}{\mathbb{F}}

\newcommand{\fa}{\frak{a}}

\newcommand{\fb}{\frak{b}}

\newcommand{\fp}{\frak{p}}

\begin{document}
\author[Z. Sepasdar]
{Zahra Sepasdar}

\title[quasi cyclic codes]
{Generator polynomials and generator matrix for quasi cyclic codes}
\thanks{E-mail addresses:
zahra.sepasdar@mail.um.ac.ir and zahra.sepasdar@gmail.com}
\maketitle

\begin{center}
{\it
 Department of Pure Mathematics, Ferdowsi University of Mashhad,\\
P.O.Box 1159-91775, Mashhad, Iran} \\
\end{center}
\vspace{0.4cm}
\begin{abstract}
Quasi-cyclic (QC) codes form an important generalization of cyclic codes. It is well know that QC codes of length $s\ell$ with index $s$ over the finite field $\F$ are $\F[y]$-submodules of the ring $\F[x,y]/{<x^s -1,y^{\ell}-1>}.$ 
The aim of the present paper, is to study QC codes of length $s\ell$ with index $s$ over the finite field $\F$ and  find generator polynomials and generator matrix for these codes. To achieve this aim, we apply a novel method to find generator polynomials for $\F[y]$-submodules of  $\F[x,y]/{<x^s -1,y^{\ell}-1>}$. These polynomials will be applied to obtain generator matrix for corresponding QC codes. 
\vspace{0.5cm}
\end{abstract}

\section{Introduction}
Cyclic codes form an important family of codes. Cyclic codes of length $n$ over the finite field $\F$ are ideals of the polynomial ring $\F[x]/<x^n-1>$. It is well known that $\pi: \F^n\longrightarrow \F[x]/<x^n-1>$ which takes an element $a=(a_0,a_1,\dots,a_{n-1}) \in \F^n$ to the polynomial $a(x)=a_0+a_1x+\dots+a_{n-1}x^{n-1}$, is a one to one correspondence between cyclic codes over $\F$ and ideals of $\F[x]/<x^n-1>$.
Obviously, $\pi$ can be generalized to $$\mu: \F^{s\ell} \longrightarrow \F[x,y]/<x^s-1,y^{\ell}-1>$$ which  takes an element $$a=(a_{0,0},a_{0,1},\dots,a_{0,s-1},a_{1,0},a_{1,1},\dots,a_{1,s-1},\dots, a_{\ell-1,0},a_{\ell-1,1},\dots,a_{\ell-1,s-1}) \in \F^{s\ell}$$ to a bivariate polynomial of degree at most $s-1$ with respect to $x$ and $\ell-1$ with respect to $y$ 
 \begin{align*}
a(x,y)&=(a_{0,0}+a_{0,1}x+\dots+a_{0,s-1}x^{s-1})y^0+(a_{1,0}+a_{1,1}x+\dots+a_{1,s-1}x^{s-1})y \\&+\dots+
(a_{\ell-1,0}+a_{\ell-1,1}x+\dots +a_{\ell-1,s-1}x^{s-1})y^{\ell-1}.
 \end{align*}
In fact, $\mu$ is a correspondence between one of the generalization of cyclic codes, two-dimensional cyclic (TDC) codes of length $n=\ell s$ and ideals of the polynomial ring $\F[x,y]/{<x^s -1,y^\ell-1>}.$ 
The algebraic structure of some of these codes  and their dual codes were studied by the present author in \cite{zahra} and their generator polynomials are obtained. Another important generalization of cyclic code is quasi-cyclic code. 
\begin{definition}
A linear code $C$ of length $n=s\ell$ over the finite field $\F$ is called a quasi-cyclic code (QC code) of index $s$ if for every codeword $c\in C$ the codeword obtained by $s$-cyclic shifts is also a codeword in $C$. That is,
$$  (c_0, c_1, \dots, c_{n-1}) \in C  \Longrightarrow   (c_{n-s},\dots, c_0, \dots ,c_{n-s-1}) \in C.$$
\end{definition}
According to the definition of $\mu$, QC codes of length $s\ell$ with index $s$ over the finite field $\F$ are $\F[y]$-submodules of the ring $\F[x,y]/{<x^s -1,y^{\ell}-1>}.$
QC codes form an interesting family of codes, since many codes with best minimum distance are QC codes. Many coding theorists have studied the structure of these codes, see \cite{G1, G2}. And some authors tried to characterize these codes. In \cite{lally}, Lally et al. used the language of Gr\"obner bases to characterize the structure of QC codes. Also in \cite{G4, G5} Ling et al. investigated the algebraic structures of QC codes. 

One of the main concerns about QC codes is finding the related generator
polynomials, because this enables us to investigate the structure of QC codes. This procedure most probably helps to decode QC codes as it has been done for cyclic codes. The aim of the present paper, is to study QC codes of length $s\ell$ with index $s$ over the finite field $\F$ and find generator polynomials and generator matrix for them. Since QC codes of length $s\ell$ with index $s$ are $\F[y]$-submodules of the ring $\F[x,y]/{<x^s -1,y^{\ell}-1>},$ we study the algebraic structure of  $\F[y]$-submodules of the ring $\F[x,y]/{<x^s -1,y^{\ell}-1>}$. To achieve this purpose, we apply a novel method to find generator polynomials for $\F[y]$-submodules of $\F[x,y]/{<x^s -1,y^{\ell}-1>}$ and then use these polynomials to obtain a generator matrix for corresponding QC codes. It is worth to note that similar method can be applied to obtain generator matrix for the other families of codes, such as generalized quasi-cyclic codes and multi-twisted codes.

\begin{remark}
For simplicity of notation, we write $g(x)$ instead of $g(x)+<\fa>$ for elements of $\F[x]/<\fa>$. Similarly, we write $g(x,y)$ instead of $g(x,y)+<\fa,\fb>$ for elements of $\F[x,y]/<\fa,\fb>$.
\end{remark}

\section{Generator polynomials}
Set $R:=\F[x,y]/<x^s-1, y^{\ell}-1>$ and $S:=\F[y]/<y^{\ell}-1>$. Suppose that $M$ is an $\F[y]$-submodule of $R$. In this section, we construct ideals $I_i$  of $S$ ($i=0,\dots, s-1$) and prove that generator polynomials of these ideals provide a generating set for $M$. Assume that $f(x,y)$ is an arbitrary element  of $M$. Since
$$\F[x,y]/<x^s-1,y^{\ell}-1>\cong (\F[y]/<y^{\ell}-1>)[x]/<x^s-1>,$$
$f(x,y)$ can be written uniquely as $f(x,y)=\sum_{i=0}^{s-1}f_i(y)x^i$, where $f_i(y) \in S$ for $i=0,\dots, s-1$.
Put
\begin{align*}
I_0=\{g_0(y) \in S: \T{there \  exists \ } & g(x,y)\in M \ \T{such \  that}\ g(x,y)=\sum_{i=0}^{s-1}g_i(y)x^i\}.
\end{align*}
First, we prove that $I_0$ is an ideal of the ring $S$. Assume that $g_0(y)$ is an arbitrary element of $I_0$. According to the definition of $I_0$, there exists $g(x,y)\in M$ such that $g(x,y)=\sum_{i=0}^{s-1}g_i(y)x^i$. Now,  $yg_0(y)\in I_0$ since $M$ is an $\F[y]$-submodule of $R$ and $yg(x,y)=\sum_{i=0}^{s-1}yg_i(y)x^i$ is an element of $M$. Also it is clear that $I_0$ is closed under addition, so $I_0$ is an ideal of $S$.
It is well known that $S$ is a principal ideal ring. Thus, there exists a unique monic polynomial $p_0^{0}(y)$ in $S$ such that $I_0=<p_0^{0}(y)>$ and $p_0^{0}(y)$ is a divisor of $y^{\ell}-1$. So there exists a polynomial $p'_0(y)$ in $\F[y]$ such that $y^{\ell}-1=p'_0(y)p_0^{0}(y)$. Since $f(x, y) \in M$ and according to the definition of $I_0$, $f_0(y) \in I_0=<p_0^{0}(y)>$. So 
\begin{align}\label{1}
f_0(y)=p_0^{0}(y) q_0(y)
\end{align}
for some $q_0 (y) \in S$. 
Now,  $p_0^{0}(y) \in I_0$  so according to the definition of $I_0$, there exists $\fp_0(x,y) \in M$ such that $$\fp_0 (x,y)=\sum_{i=0}^{s-1} p_i^{0}(y)x^i.$$ Set 
\begin{align*}
h_1(x,y):&=f(x,y)-\fp_0(x,y)q_0(y)=\sum_{i=0}^{s-1}f_i(y)x^i-q_0(y)\sum_{i=0}^{s-1} p_i^{0}(y)x^i\\&=
f_0(y)+\sum_{i=1}^{s-1}f_i(y)x^i-p_0^{0}(y) q_0(y)-q_0(y)\sum_{i=1}^{s-1} p_i^{0}(y)x^i\\&=\sum_{i=1}^{s-1}f_i(y)x^i-q_0(y)\sum_{i=1}^{s-1} p_i^{0}(y)x^i. \ \ \ \ \ \  \ \ \ \ \  \ \ \ \ \ \T{(by\ equation\ \ref{1})} 
\end{align*}
Since $f(x,y)$ and $\fp_0(x,y)$ are in $M$ and $M$ is an $\F[y]$-submodule of $R$, $h_1(x,y)$ is a polynomial in $M$. Also note that  $h_1(x,y)$ is in the form of $h_1(x,y)=\sum_{i=1}^{s-1} h_i^{1}(y)x^i$ for some $h_i^{1}(y) \in S$. Put
\begin{align*}
I_1=\{g_1(y) \in S:& \ \T{there \  exists} \ g(x,y)\in M \ \T{such \  that \ }  g(x,y)=\sum_{i=1}^{s-1}g_i(y)x^i\}.
\end{align*}
By the same method we applied for $I_0$, it can be proved that $I_1$ is an ideal of $S$. Thus, there exists a unique monic polynomial $p_1^1(y)$ in $S$ such that $I_1=<p_1^{1}(y)>$ and $p_1^{1}(y)$ is a divisor of $y^{\ell}-1$. Therefore, there exists a polynomial $p'_1(y)$ in $\F[y]$ such that $y^{\ell}-1=p'_1(y)p_1^{1}(y)$. Now, $h_1(x,y) \in M$ so according to the definition of $I_1$, $h_1^{1}(y) \in I_1=<p_1^{1}(y)>$. So 
\begin{align}\label{2}
h_1^{1}(y)=p_1^{1}(y) q_1(y)
\end{align}
 for some $ q_1(y) \in S$. 
Also, $p_1^1(y) \in I_1$ so according to the definition of $I_1$, there exists $\fp_1(x,y) \in M$ such that $$\fp_1(x,y)=\sum_{i=1}^{s-1} p_i^1(y)x^i.$$
Set
\begin{align*}
h_2(x,y):&=h_1(x,y)-\fp_1(x,y)q_1(y)=\sum_{i=1}^{s-1} h_i^{1}(y)x^i-q_1(y)\sum_{i=1}^{s-1} p_i^1(y)x^i\\&=h_1^{1}(y)x+\sum_{i=2}^{s-1} h_i^{1}(y)x^i-p_1^{1}(y) q_1(y)x-q_1(y)\sum_{i=2}^{s-1} p_i^1(y)x^i\\&=\sum_{i=2}^{s-1} h_i^{1}(y)x^i-q_1(y)\sum_{i=2}^{s-1} p_i^1(y)x^i.\ \ \ \ \ \  \ \ \ \ \  \ \ \ \ \ \T{(by\ equation\ \ref{2})} 
\end{align*}
Since $h_1(x,y)$ and $\fp_1(x,y)$ are in $M$ and $M$ is an $\F[y]$-submodule of $R$, $h_2(x,y)$ is a polynomial in $M$ in the form of $h_2(x,y)=\sum_{i=2}^{s-1}h_i^{2}(y)x^i$ for some $h_i^{2}(y) \in S$.
Put
\begin{align*}
I_2=\{g_2(y) \in S: & \ \T{there \  exists \ }  g(x,y)\in M \ \T{such \  that \ } g(x,y)= \sum_{i=2}^{s-1} g_i(y)x^i\}.
\end{align*}
Again $I_2$ is an ideal of $S$, and so there exists a unique monic polynomial $p_2^{2}(y)$ in $S$ such that $I_2=<p_2^{2}(y)>$. Also  $p_2^{2}(y)$ is a divisor of $y^{\ell}-1$, and so there exists a polynomial $p'_2(y)$ in $\F[y]$ such that $y^{\ell}-1=p'_2(y)p_2^{2}(y)$. Now, $h_2(x,y) \in M$ so according to the definition of $I_2$, $h_2^{2}(y) \in I_2=<p_2^{2}(y)>$. So 
\begin{align}\label{3}
h_2^{2}(y)=p_2^{2}(y)q_2(y)
\end{align}
for some $q_2(y)\in S$. 
Also, $p_2^2(y) \in I_2$ so according to the definition of $I_2$, there exists $\fp_2(x,y) \in M$ such that $$\fp_2(x,y)=\sum_{i=2}^{s-1}p_i^2(y)x^i.$$
Set 
\begin{align*}
h_3(x,y):&=h_2(x,y)-\fp_2(x,y)q_2(y)=\sum_{i=2}^{s-1} h_i^{2}(y)x^i-q_2(y)\sum_{i=2}^{s-1} p_i^2(y)x^i\\&=h_2^{2}(y)x^2+\sum_{i=3}^{s-1} h_i^{2}(y)x^i-p_2^{2}(y) q_2(y)x^2-q_2(y)\sum_{i=3}^{s-1} p_i^2(y)x^i\\&=\sum_{i=3}^{s-1} h_i^{2}(y)x^i-q_2(y)\sum_{i=3}^{s-1} p_i^2(y)x^i.\ \ \ \ \ \  \ \ \ \ \  \ \ \ \ \ \T{(by\ equation \ \ref{3})} 
\end{align*}
 Thus, $h_3(x,y)$ is a polynomial in $M$ in the form of $h_3(x,y)=\sum_{i=3}^{s-1}h_i^3(y)x^i$ for some $h_i^{3}(y) \in S$.
Put
\begin{align*}
I_3=\{&g_3(y) \in S: \T{there \  exists \ } g(x,y)\in M \ \T{such \  that \ } g(x,y)=\sum_{i=3}^{s-1}g_i(y)x^i\}.
\end{align*}
Clearly $I_3$ is an ideal of $S$. So there exists a unique monic polynomial $p_3^{3}(y)$ in $S$ such that $I_3=<p_3^{3}(y)>$ and $p_3^{3}(y)$ is a divisor of $y^{\ell}-1$. Thus, there exists a polynomial $p'_3(y)$ in $\F[y]$ such that $y^{\ell}-1=p'_3(y)p_3^{3}(y)$.
 Now, $h_3(x,y) \in M$ so according to the definition of $I_3$,  $h_3^{3}(y) \in I_3=<p_3^{3}(y)>$. Therefore,
\begin{align}\label{4}
h_3^{3}(y)=p_3^{3}(y)q_3(y)
\end{align}
for some $q_3(y)\in S$. 
Also,  $p_3^3(y) \in I_3$ so according to the definition of $I_3$, there exists $\fp_3(x,y) \in M$ such that $$\fp_3(x,y)=\sum_{i=3}^{s-1}p_i^3(y)x^i.$$ Set
\begin{align*}
h_4(x,y):&=h_3(x,y)-\fp_3(x,y)q_3(y)=\sum_{i=3}^{s-1} h_i^{3}(y)x^i-q_3(y)\sum_{i=3}^{s-1} p_i^3(y)x^i\\&=h_3^{3}(y)x^3+\sum_{i=4}^{s-1} h_i^{3}(y)x^i-p_3^{3}(y) q_3(y)x^3-q_3(y)\sum_{i=4}^{s-1} p_i^3(y)x^i\\&=\sum_{i=4}^{s-1} h_i^{3}(y)x^i-q_3(y)\sum_{i=4}^{s-1} p_i^3(y)x^i.\ \ \ \ \ \  \ \ \ \ \  \ \ \ \ \ \T{(by\ equation\ \ref{4})} 
\end{align*}
So $h_4(x,y)$ is a polynomial in $M$ in the form of $h_4(x,y)=\sum_{i=4}^{s-1}h_i^4(y)x^i$ for some $h_i^{4}(y) \in S$.
In the next step, we put
\begin{align*}
I_4=\{& g_4(y) \in S: \T{there \  exists \ }  g(x,y)\in M \  \T{such \  that \ } g(x,y)=\sum_{i=4}^{s-1} g_i(y)x^i\}.
\end{align*}
The same procedure is applied to obtain 
$$h_5(x,y),\dots,h_{s-2}(x,y),\fp_4(x,y),\dots,\fp_{s-2}(x,y)$$
in $M$ and
$q_4(y),\dots,q_{s-2}(y)$ in $S$ and 
  construct ideals $I_5,\dots,I_{s-2}$. Finally, we set $$h_{s-1}(x,y):=h_{s-2}(x,y)-\fp_{s-2}(x,y)q_{s-2}(y).$$ Thus, $h_{s-1}(x,y)$ is a polynomial in $M$ in the form of $h_{s-1}(x,y)=h_{s-1}^{s-1}(y)x^{s-1}$. Set
\begin{align*}
I_{s-1}=\{ g_{s-1}(y) \in S: & \ \T{there \  exists \ }  g(x,y)\in M \ \T{such \  that \ }  g(x,y)=g_{s-1}(y)x^{s-1}\}.
\end{align*}
Clearly $I_{s-1}$ is an ideal of $S$. Thus, there exists a unique monic polynomial $p_{s-1}^{s-1}(y)$ in $S$ such that $I_{s-1}=<p_{s-1}^{s-1}(y)>$ and $p_{s-1}^{s-1}(y)$ is a divisor of $y^{\ell}-1$ (there exists $p'_{s-1}(y)$ in $\F[y]$ such that $y^{\ell}-1=p'_{s-1}(y)p_{s-1}^{s-1}(y)$). Now, $h_{s-1}(x,y) \in M$ so according to the definition of $I_{s-1}$,  $h_{s-1}^{s-1}(y) \in I_{s-1}=<p_{s-1}^{s-1}(y)>$. So
\begin{align}\label{kolo}
h_{s-1}^{s-1}(y)=q_{s-1}(y)p_{s-1}^{s-1}(y)
\end{align}
for some $q_{s-1}(y) \in S$. 
Also, $p_{s-1}^{s-1}(y) \in I_{s-1}$ so by definition of $I_{s-1}$, there exists $\fp_{s-1}(x,y) \in M$ such that $\fp_{s-1}(x,y)=p_{s-1}^{s-1}(y)x^{s-1}.$ According to the equation \ref{kolo}, 
$$h_{s-1}(x,y)=h_{s-1}^{s-1}(y)x^{s-1}=q_{s-1}(y)p_{s-1}^{s-1}(y)x^{s-1}=q_{s-1}(y)\fp_{s-1}(x,y).$$
Thus, for an arbitrary element $f(x,y) \in M$ we show that
\begin{align*}
& h_1(x,y):=f(x,y)-\fp_0(x,y)q_0(y)\\&
h_2(x,y):=h_1(x,y)-\fp_1(x,y)q_1(y)\\&
h_3(x,y):=h_2(x,y)-\fp_2(x,y)q_2(y)\\&
h_4(x,y):=h_3(x,y)-\fp_3(x,y)q_3(y)\\&
\dots\\&
h_{s-1}(x,y):=h_{s-2}(x,y)-\fp_{s-2}(x,y)q_{s-2}(y)\\&
h_{s-1}(x,y)=q_{s-1}(y)\fp_{s-1}(x,y).
\end{align*}
So 
\begin{align*}
f(x,y)&=\fp_0(x,y)q_0(y)+\fp_1(x,y)q_1(y)+\fp_2(x,y)q_2(y)+\fp_3(x,y)q_3(y)\\&+
\dots+\fp_{s-2}(x,y)q_{s-2}(y)+\fp_{s-1}(x,y)q_{s-1}(y).
\end{align*}
Since $\fp_{i}(x,y)\in M$ for $i=0,\dots,s-1$ and $f(x,y)$ is an arbitrary element of $M$ and $M$ is an $\F[y]$-submodule of $R$, we conclude that 
$$M=<\fp_0(x,y),\dots,\fp_{s-1}(x,y)>.$$
So $\{\fp_0(x,y), \fp_1(x,y),\dots,\fp_{s-1}(x,y)\}$ is a set of  generating polynomials for $M$. 

In the next theorem, we introduce the generator matrix for QC codes.
\begin{theorem}
Suppose that $M$ is an $\F[y]$-submodule of $\F[x,y]/<x^s-1,y^{\ell}-1>$ and is generated by $\{\fp_0(x,y),\dots,\fp_{s-1}(x,y)\}$, which obtained from the above method. Then the set
\begin{align*}
\{&\fp_0(x,y),y\fp_0(x,y),\dots,y^{\ell-a_0-1}\fp_0(x,y), \\ & \fp_1(x,y),y\fp_1(x,y),\dots,y^{\ell-a_1-1}\fp_1(x,y), \\ & \ \ \ \ \ \ \ \ \ \ \ \ \ \ \ \ \ \ \ \  \ \ \ \ \ \vdots \\ & \fp_{s-1}(x,y),y\fp_{s-1}(x,y),\dots,y^{\ell-a_{s-1}-1}\fp_{s-1}(x,y)\}
\end{align*}
forms an $\F$-basis for $M$, where  $a_i=\T{deg}(p_{i}^i(y))$.
\end{theorem}
\begin{proof}
Assume that $l_0(y),\dots, l_{s-1}(y)$ are polynomials in $\F[y]$ such that $\T{deg}(l_i(y))< \ell-a_i$ and $l_0(y) \fp_0(x,y)+\dots+l_{s-1}(y)\fp_{s-1}(x,y)=0$. These imply the following equation in $S$ $$l_0(y)  p_0^0(y)=0.$$ Therefore, 
$l_0(y) p_0^0(y)=s(y) (y^{\ell}-1)$ for some $s(y) \in \F[y]$. Now,  the degree of $y$ in the right side of this equation is at least $\ell$ but since $\T{deg}(p_0^0(y))=a_0$ and $\T{deg}(l_0(y))< \ell-a_0$, the degree of $y$ in the left side of this equation is at most $\ell-1$. So we get
$l_0(y)=0$. Similar arguments yield $l_i(y)=0$ for $i=1,\dots,s-1$.
\end{proof}

\section{Conclusion}
In this paper, we investigate the structure of quasi cyclic codes of length $s\ell$ with index $s$. This leads to studying the structure of $\F[y]$-submodules of the ring $\F[x,y]/<x^s-1,y^{\ell}-1>$. We apply a novel method to obtain generating sets of polynomials  for $\F[y]$-submodules of  $\F[x,y]/<x^s-1,y^{\ell}-1>$. These polynomials will be applied to obtain generator matrix for corresponding QC codes.
It is worth to note that similar method can be applied to obtain generator matrix for the other families of codes, such as generalized quasi-cyclic codes and multi-twisted codes.

 \end{document}